\documentclass[journal]{IEEEtran}

\usepackage[T1]{fontenc}
\usepackage{xcolor}
\usepackage{mathtools, amssymb, amsfonts, dsfont}
\usepackage{microtype}
\usepackage{url}
\usepackage[sort,compress]{cite}
\usepackage[standard,amsmath,thmmarks]{ntheorem}
\usepackage{xcolor}

\let\oldproof\proof
\let\oldendproof\endproof
\let\proof\relax
\let\proof\oldproof
\let\endproof\oldendproof

\newtheorem{problem}{Problem}
\newtheorem*{fact}{Fact}

\newcommand*\reals{\mathds{R}}
\newcommand*\X[1]{\reals^{d_x^{#1}}}

\newcommand*\U[1]{\reals^{d_u^{#1}}}

\newcommand*\com{\mathrm{com}}
\newcommand*\off{\textsc{off}}

\newcommand*\DEFINED{\coloneqq}

\newcommand*\EXP{\mathds{E}\strut}
\newcommand*\PR{\mathds{P}}
\newcommand*\MATRIX[1]{\begin{bmatrix}#1\end{bmatrix}}

\newcommand*\BLANK{\mathfrak E}
\DeclareMathOperator{\Riccati}{\mathcal{R}}
\DeclareMathOperator{\Gain}{\mathcal{G}}

\newcommand\independent{\protect\mathpalette{\protect\independenT}{\perp}}
\def\independenT#1#2{\mathrel{\rlap{$#1#2$}\mkern2mu{#1#2}}}

\DeclareMathOperator\VVEC{vec}

\DeclareMathOperator\ROWS{rows}

\newcommand*\TRANS{{\mathpalette\doTRANS\empty}}
\makeatletter
\newcommand*\doTRANS[2]{\raisebox{\depth}{$\m@th#1\intercal$}}
\makeatother

\begin{document}

\title{Optimal local and remote controllers with
unreliable uplink channels: An elementary proof}

\author{Mohammad~Afshari,~\IEEEmembership{Student Member,~IEEE,}
        and~Aditya~Mahajan,~\IEEEmembership{Senior Member,~IEEE}%
\thanks{This research was supported by NSERC under Discovery Accelerator Grant
493011.}%
\thanks{The authors are with the Department of Electrical and Computer
Engineering, McGill University, Montreal, QC, H3A-0E9, Canada.
Emails: {\tt\small mohammad.afshari2@mcgill.ca, aditya.mahajan@mcgill.ca}.}%
}
\maketitle


\begin{abstract}
  Recently, a model of a decentralized control system with local and remote
  controllers connected over unreliable channels was presented
  in~\cite{Asghari2018}. The model has a non-classical information structure
  that is not partially nested. Nonetheless, it is shown in~\cite{Asghari2018}
  that the optimal control strategies are linear functions of the state estimate 
	(which is a non-linear function of the observations). 
	Their proof is based on a
  fairly sophisticated dynamic programming argument. In this note, we present
  an alternative and elementary proof of the result which uses common
  information based conditional independence and completion of squares.
\end{abstract}


\section{Introduction}
In a recent paper, a methodology for synthesizing optimal control laws for
local and remote controllers for networked control of a linear system over
unreliable uplink channel was presented~\cite{Asghari2018}. Such models arise
in applications such as temperature control in smart buildings, control of
UAVs, vehicle to infrastructure communication, etc.

The model proposed in~\cite{Asghari2018} is a decentralized control system
with non-classical information structure. Due to the unreliable nature of the
uplink channels, the information structure is not partially nested. Therefore,
one cannot a priori restrict attention to linear strategies. Nonetheless, it
is shown in~\cite{Asghari2018} that the optimal local and remote control laws are
linear functions of the state estimate 	(which is a non-linear 
function of the observations). See Theorem~\ref{thm:optimal} for a precise 
statement of the result.

The proof technique employed in~\cite{Asghari2018} uses ideas from the common
information approach of~\cite{Nayyar2013} to compute the optimal control laws.
Using a conditional independence argument, it is first shown that the local
controllers can ignore the past realization of their local states without any
loss of optimality~\cite[Lemma 1]{Asghari2018}. When attention is restricted
to control strategies with such a structure the resulting information
structure is partial history sharing~\cite{Nayyar2013}. So, in principle, the
common information approach of~\cite{Nayyar2013} is applicable. However, there
are several technical difficulties in extending the argument given
in~\cite{Nayyar2013} for finite valued random variables to continuous random
variables. The key result of~\cite{Asghari2018} is to carefully resolve these
technical difficulties---issues of measurability, existence of well defined
value function, and infinite dimensional strategy space---and then obtain a
closed form solution of the dynamic program.

In this technical note, we provide an alternative and elementary
proof of the result of~\cite{Asghari2018}. Our proof also relies on the
split of total information into common and local information as proposed
in~\cite{Nayyar2013}. However, instead of using the dynamic program
proposed in~\cite{Nayyar2013}, we develop an alternative solution
methodology which relies on (i)~the conditional independence of the local
states given the common information (which was established
in~\cite{Asghari2018}); and (ii)~simplifying the per-step cost based on this
conditional independence, orthogonality principle, and completion of squares. 
The key advantage of this
solution approach is that it completely sidesteps the technical difficulties
with measurability and existence of value functions present in a dynamic
programming based approach.
Given the paucity of positive results in optimal control of decentralized
systems, we believe that a new solution approach is interesting in its own
right.

The model considered in~\cite{Asghari2018} consists of $N$ local controllers
and one remote controller. For ease of exposition, we assume that $N = 2$. It
will be clear from the proof that the steps extend to general
$N$. For the most part, we broadly follow the notation and terminology
of~\cite{Asghari2018}, but we occasionally deviate from it to be consistent
with the standard notation used in linear systems.

\subsection{Notations}

We use superscripts to indicate subsystems/controllers and subscripts to
indicate time. Thus, $x^i_t$ denotes the state of subsystem~$i$ at time~$t$.
The superscript ${}^\TRANS$ denotes transpose (of a vector or a matrix).
$\mathbf{0}_{m \times n}$ is a $m \times n$ matrix with all elements being
equal to zero. We omit the subscript from $\mathbf{0}_{m \times n}$ when 
the dimension is clear from context.
Given column vectors $x$ and $y$, the notation $\VVEC(x,y)$ is a short hand
for the vector formed by stacking $x$ on top of $y$. Given
random variables $x$, $y$, and $z$, the notation $x \independent y \mid z$
indicates that $x$ and $y$ are conditionally independent given $z$. Given
matrices $A$ and $B$ with the same number of columns, $\ROWS(A,B)$ denotes
the matrix obtained by stacking $A$ on top of $B$.

Given matrices $A$, $B$, $Q$, $M$, $R$, and $P$ of appropriate dimensions, we
use the following operators:
 \begin{multline*}
   \Riccati(P, A, B, Q, M, R) = Q + A^\TRANS P A \\
    - (M + A^\TRANS P B)(R + B^\TRANS P B)^{-1}(M + A^\TRANS P B)^\TRANS ,
 \end{multline*}
 \begin{equation*}
   \Gain(P, A, B, M, R) = - (R + B^\TRANS P B)^{-1} (M + A^\TRANS P B)^\TRANS ,
 \end{equation*}
which denote the one step update of the discrete time Riccati equation and the
gain of a linear system, respectively.

\section{Model and Problem Formulation}

\subsection{System dynamics}

Consider a discrete-time linear dynamical system consisting of $N = 2$
subsystems. $x^i_t \in \X{i}$ denotes the state of subsystem~$i$, $i \in \{1,
2\}$. There is a local controller $C^i$ co-located with subsystem~$i$. In
addition, there is a remote controller $C^0$. The information available to
the controllers will be described later. Let $u^i_t \in \U{i}$, $i \in
\{1, 2\}$, denote the control action of local controller~$C^i$ and $u^0_t \in
\U{0}$ denote the control action of remote controller~$C^0$.

The initial state $x^i_0$ of subsystem~$i$, $i \in \{1,2\}$, is random
and the dynamics of subsystem~$i$ is given by
\begin{equation} \label{eq:dynamics}
  x^i_{t+1} = A^{ii} x^i_t +
  \MATRIX{B^{i0} & B^{ii}} \MATRIX{u^0_t \\ u^i_t} +
  w^i_t,
\end{equation}
where $w^i_t \in \X{i}$ is the process noise and $A^{ii}$, $B^{i0}$, and
$B^{ii}$ are matrices of appropriate dimensions. We assume that random variables 
$\{w^1_0, \dots, w^1_{T-1}, w^2_0, \dots,w^2_{T-1} \}$ are independent and have 
zero mean and finite variance. Let $x_t \DEFINED
\VVEC(x^1_t, x^2_t)$, $u_t \DEFINED \VVEC(u^0_t, u^1_t, u^2_t)$, 
and $w_t \DEFINED \VVEC(w^1_t, w^2_t)$ denote the state, control actions, and noise of the
overall system. Then, the system dynamics can be written as
\begin{equation}
  x_{t+1} = A x_t + B u_t + w_t,
\end{equation}
where the matrices $A$ and $B$ are given by
\[
  A = \MATRIX{ A^{11} & 0 \\ 0 & A^{22} }
  \quad\text{and}\quad
  B = \MATRIX{ B^{10} & B^{11} & 0 \\ B^{20} & 0 & B^{22} }.
\]

\subsection{Information structure}

At time~$t$, the local controller $C^i$, $i \in \{1, 2\}$, perfectly observes
the state $x^i_t$ of subsystem~$i$ and sends it to the remote controller $C^0$
over an unreliable packet drop channel. Let $\Gamma^i_t \in \{0, 1\}$ denote
the state of the channel, where $\Gamma^i_t = 0$ means that the channel is in
the off state where the transmitted packet gets dropped while $\Gamma^i_t = 1$
means that the channel is in the on state where the transmitted packet gets
delivered. Thus, $\Gamma^i_t$ is a Bernoulli random variable and we denote the
packet drop probability $\PR(\Gamma^i_t = 0)$ by $p^i$. We use $\Gamma_t$ to
denote $(\Gamma^1_t, \Gamma^2_t)$.

Let $z^i_t$ denote the output of the channel~$i$, $i \in \{1, 2\}$, i.e.,
\begin{equation} \label{eq:channel}
  z^i_t = f(x^i_t, \Gamma^i_t) = \begin{cases}
    x^i_t, & \hbox{if $\Gamma^i_t = 1$} \\
    \BLANK, & \hbox{if $\Gamma^i_t = 0$}
  \end{cases}
\end{equation}
where $\BLANK$ denotes a dropped packet. It is assumed that there are perfect
channels from $C^0$ to $C^1$ and $C^2$. Using these channels, $C^0$ can share
$z_t \DEFINED \VVEC(z^1_t, z^2_t)$ and $u^0_{t-1}$ with local controllers $C^1$ and
$C^2$. Note that it is possible to recover $\Gamma^i_t$ from $z^i_t$. Hence, all
controllers also have access to $\Gamma_t$. The fact that $\Gamma_t$ is available 
at all controllers is critical to derive the main result of the model (presented 
in Theorem~\ref{thm:optimal}).

Let $H^i_t$, $i \in \{0,1,2\}$, denote the information available to controller
$C^i$ to take decisions at time~$t$. Then,
\begin{subequations} \label{eq:info-structure}
\begin{align}
  H^0_t &= \{ z_{0:t}, \Gamma_{0:t}, u^0_{0:t-1}  \}, \\
  H^i_t &= \{ x^i_{0:t}, u^i_{0:t-1}, z_{0:t}, \Gamma_{0:t}, u^0_{0:t-1} \},
  \quad i \in \{1, 2\}.
\end{align}
\end{subequations}
Let $\mathcal{H}^i_t$ be the space of all possible realizations of $H^i_t$.
Then, controller $C^i$ chooses it's control action according to
\begin{equation} \label{eq:control}
  u^i_t = g^i_t(H^i_t), \quad i \in \{0,1,2\},
\end{equation}
where the Borel measurable function $g^i_t \colon \mathcal{H}^i_t \rightarrow
\U{i}$ is called the \emph{control law of controller $C^i$} at time~$t$. The
collection $\mathbf g^i = (g^i_0, \dots, g^i_T)$ is called the \emph{control
strategy of controller $C^i$}. The collection $\mathbf g \DEFINED (\mathbf
g^0, \mathbf g^1, \mathbf g^2)$ is called the \emph{strategy profile of the
system}.

\subsection{System performance and the optimization problem}

The system operates for a finite horizon~$T$. For time $t < T$, the system
incurs a per-step cost
\[
  c_t(x_t, u_t)
  = \MATRIX{ x_t \\ u_t}^\TRANS
  \MATRIX{ Q_t & M_t \\ M_t^\TRANS & R_t}
  \MATRIX{ x_t \\ u_t}
\]
and for the terminal time~$T$, the system incurs a terminal cost
\[
  c_T(x_T) = x_T^\TRANS Q_T x_T,
\]
where $Q_t$, $M_t$, and $R_t$ are matrices of appropriate dimensions. We
assume the following block-wise structure of $Q_t$, $M_t$, and~$R_t$:
\[
  Q_t = \MATRIX{ Q^{11}_t & Q^{12}_t \\ Q^{21}_t & Q^{22}_t },
  \quad
  M_t = \MATRIX{ M^{10}_t & M^{11}_t & M^{12}_t \\ M^{20}_t & M^{21}_t & M^{22}_t },
\]
and
\[
  R_t = \MATRIX{ R^{00}_t & R^{01}_t & R^{02}_t \\
  R^{10}_t & R^{11}_t & R^{12}_t \\ R^{20}_t & R^{21}_t & R^{22}_t }.
\]

The performance of a strategy profile $\mathbf g$ is given by
\begin{equation} \label{eq:cost}
  J(\mathbf g) =
  \EXP^{\mathbf g}
  \Big[ \sum_{t=0}^{T-1} c_t(x_t, u_t) +
  c_T(x_T) \Big],
\end{equation}
where the expectation is with respect to the measure induced on all the system
variables by the choice of strategy profile $\mathbf g$.

The following assumptions are imposed on the system:
\begin{description}
  \item[\textbf{(A1)}] The primitive random variables $\{x^1_0, x^2_0,
      \allowbreak w^1_0, \dots, \allowbreak w^1_{T-1},
      \allowbreak w^2_0, \dots, \allowbreak w^2_{T-1},
      \allowbreak \Gamma^1_0, \dots, \Gamma^1_{T-1},
      \allowbreak \Gamma^2_0, \dots, \Gamma^2_{T-1} \}$ are independent.
    \item[\textbf{(A2)}] The variables $\{x^1_0, x^2_0, w^1_0, \dots,
      w^1_{T-1}, w^2_0, \dots, w^2_{T-1} \}$ have zero mean and finite
      variance. We use $\Sigma^i_t$ and $\Sigma^i_x$ to denote the variance
      of $w^i_t$ and $x^i_0$ respectively.
    \item[\textbf{(A3)}] For each $t$,
      the matrix $\left[ \begin{smallmatrix} Q_t & M_t \\
      M_t^\TRANS & R_t \end{smallmatrix} \right]$ is symmetric and positive
      semi-definite, and the matrix $R_t$ is symmetric and positive definite.
\end{description}

We are interested in the following optimization problem.

\begin{problem} \label{prob:main}
  In the model described above, find a strategy profile $\mathbf g^* =
  (\mathbf g^{*,0}, \mathbf g^{*,1}, \mathbf g^{*,2})$ that
  minimizes~\eqref{eq:cost}, i.e.,
  \[
    J^* \DEFINED J(\mathbf g^*) =
    \inf_{\mathbf g} J(\mathbf g),
  \]
  where the infimum is taken over all strategy profiles of the
  form~\eqref{eq:control}.
\end{problem}

\subsection{Some remarks}

The per-step cost function defined above differs slightly from the per-step
cost function considered in~\cite{Asghari2018} in the following ways:
\begin{itemize}
  \item In~\cite{Asghari2018}, the matrix
    $\left[ \begin{smallmatrix} Q_t & M_t \\ M_t^\TRANS & R_t \end{smallmatrix} \right]$
    was denoted by $R_t$.
    We follow the standard notation here.

  \item In~\cite{Asghari2018}, it was assumed that the performance of a
    strategy profile is
    \[
      \EXP^{\mathbf g}
      \Big[ \sum_{t=0}^{T} c_t(x_t, u_t) \Big].
    \]
    This is effectively the same as assuming that there is no terminal cost (i.e.,
    $Q_{T+1} = 0$) and therefore the terminal control actions $u^i_T$ are $0$ for both local and remote
    controllers. To avoid such triviality, we assume a performance function of
    the form~\eqref{eq:cost}.
\end{itemize}

\section{Main result}

In this section, we restate the main results of~\cite{Asghari2018} but we
present them in a slightly different manner.

\subsection{Common information based estimates}

Following~\cite{Nayyar2013}, we define the \emph{common information}
$H^\com_t$ between agents as
\[
  H^\com_t = H^0_t \cap H^1_t \cap H^2_t.
\]
The information structure of the model~\eqref{eq:info-structure} implies that
$H^\com_t = H^0_t = \{z_{0:t}, \Gamma_{0:t}, u^0_{0:t-1} \}$.

Now we define the common information based ``estimates'' of the state and
control actions and the corresponding ``estimation errors'' as follows:
\begin{align}
  \hat x_t &= \EXP[ x_t \mid H^\com_t],
  & \tilde x_t = x_t - \hat x_t ,
  \label{eq:x}
  \\
  \hat u_t &= \EXP[ u_t \mid H^\com_t],
  & \tilde u_t = u_t - \hat u_t .
  \label{eq:u}
\end{align}

For ease of notation, we use $\hat x^i_t$
to denote the $i_{th}$ component of $\hat x_t$, i.e., $\hat x_t = \VVEC(\hat x^1_t
, \hat x^2_t)$. Similar interpretation holds for $\tilde x^i_t$, $\hat u^i_t$, and 
$\tilde u^i_t$.

It can be shown that the state estimates and the estimation error satisfy the
following property.
\begin{lemma}  \label{lem:estimates}
  The state estimates and estimation errors evolve as follows: for $i \in \{1,
  2\}$,
  \[
    \hat x^i_0 =
    \begin{cases}
      0, & \text{if $\Gamma^i_0 = 0$} \\
      x^i_0, & \text{if $\Gamma^i_0 = 1$}
    \end{cases}
  \]
  and for $t > 0$,
  \[
    \hat x^i_{t+1} =
    \begin{cases}
      A^{ii} \hat x^i_t + B^{i0} u^0_t + B^{ii} \hat u^i_t,
      & \mbox{if $\Gamma^i_{t+1} = 0$}\\
      x^i_{t+1}, &\mbox{if $\Gamma^i_{t+1} = 1$}.
    \end{cases}
  \]
  Therefore,
  \[
    \tilde x^i_0 =
    \begin{cases}
      x^i_0, & \text{if $\Gamma^i_0 = 0$} \\
      0, & \text{if $\Gamma^i_0 = 1$}
    \end{cases}
  \]
  and for $t > 0$,
  \[
    \tilde x^i_{t+1} =
    \begin{cases}
      A^{ii} \tilde x^i_t + B^{ii} \tilde u^i_t + w^i_t,
      & \mbox{if $\Gamma^i_{t+1} = 0$}\\
      0, &\mbox{if $\Gamma^i_{t+1} = 1$}.
    \end{cases}
  \]
\end{lemma}
A proof is presented in Sec.~\ref{sec:estimates}.
\begin{remark} \label{rem:info}
  Lemma~\ref{lem:estimates} along with the definition of the state and control
  estimates~\eqref{eq:x} and~\eqref{eq:u} and the information
  structure~\eqref{eq:info-structure} imply that all controllers know the
  value of $\VVEC(\hat x^1_t, \hat x^2_t)$ at time~$t$. An immediate consequence of this is that
  controller~$C^i$ knows the value of $\tilde x^i_t$ at time~$t$. The main result of the model,
	explained in the next section, is that the optimal control action at controller 
	$C^i$ is linear in $(\hat x_t, \tilde x^i_t)$.
\end{remark}
\subsection{Structure of optimal control laws}
In order to present the main result of~\cite{Asghari2018}, we recursively
define matrices $\{P_t\}_{t=1}^T$ as follows: $P_T = Q_T$ and for $t \in
\{T-1, \dots, 1\}$,
\begin{equation} \label{eq:P}
  P_t = \Riccati(P_{t+1}, A, B, Q_t, M_t, R_t).
\end{equation}
Furthermore, let $P^{ii}_t$ denote the $(i,i)$-th block of $P_t$. Then
for $i \in \{1, 2\}$, recursively define the matrices $\{\Pi^{i}_t\}_{t =
1}^T$ and $\{\tilde P^{i}_t\}_{t =
1}^T$ as follows: $\Pi^{i}_T = Q^{ii}_T$ and $\tilde P^{i}_T = Q^{ii}_T$ and for $t \in \{T-1, \dots,
1\}$, let
\begin{equation} \label{eq:tilde-P}
  \tilde P^{i}_t = \Riccati(\Pi^i_{t+1}, A^{ii}, B^{ii}, Q^{ii}_t, M^{ii}_t, R^{ii}_t).
\end{equation}
and
\begin{equation} \label{eq:Pi}
  \Pi^i_{t+1} = (1-p^i) P^{ii}_{t+1} + p^i \tilde P^i_{t+1}.
\end{equation}

The main result of~\cite{Asghari2018} is the following.
\begin{theorem}\label{thm:optimal}
  The optimal control strategy for Problem~\ref{prob:main} is
  given by
  \begin{align}
    \MATRIX{ u^0_t \\ \hat u^1_t \\ \hat u^2_t} &= -K_t \hat x_t
    \shortintertext{and}
    \tilde u^i_t &= - \tilde K^i_t \tilde x^i_t, \quad
    i \in \{1, 2\},
  \end{align}
  where the time evolution of $\hat x_t$ and $\tilde x_t$ are given by
  Lemma~\ref{lem:estimates}.

  The gains $\{K_t\}_{t=0}^{T-1}$ and $\{\tilde
  K_t\}_{t=0}^{T-1}$ are given by
  \begin{align*}
    K_t &= \Gain(P_{t+1}, A, B, M_t, R_t), \\
    \tilde K^i_t &= \Gain(\Pi^i_{t+1}, A^{ii}, B^{ii}, M^{ii}_t, R^{ii}_t), \quad i \in \{1, 2\},
  \end{align*}
where the matrices $\{P_t\}_{t=1}^{T}$, $\{\Pi^i_t\}_{t=1}^{T}$, and $\{\tilde
P^{i}\}_{t=1}^{T}$ are given by~\eqref{eq:P},~\eqref{eq:tilde-P},
and~\eqref{eq:Pi}.
\end{theorem}
\begin{remark} \label{rem:gains}
    Let $K_t = \ROWS(K^0_t, K^1_t, K^2_t)$.
    Then, Theorem~\ref{thm:optimal}
    implies that the optimal control actions are given by
   \begin{align}
	   u^0_t &= - K^0_t \hat x_t , \label{eq:equ-common} \\
	   u^i_t  &= - K^i_t \hat x_t + \tilde K^i_t (x^i_t - \hat x^i_t), 
        \quad i \in \{1, 2\}. \label{eq:equ-local}
    \end{align}
Such a control law is feasible because, $\tilde x^i_t$ is available at controller 
	$C^i$ as explained in Remark~\ref{rem:info}.

	The structure of the control laws~\eqref{eq:equ-common}-\eqref{eq:equ-local}
	implies that the optimal action is a linear function of the state estimate 
	$\hat x_t$. Note that the evolution of the state estimate, given by 
	Lemma~\ref{lem:estimates}, is a non-linear function of the data available at 
	controller $C^i$, $i \in \{1,2 \}$.
\end{remark}

\begin{remark}
	Note that the result does not depend on the distribution of the noise processes 
	$\{w^i_t\}_{t \geq 0}$, $i \in \{1,2 \}$, as long as the random variables 
	$\{w^1_0,\dots, w^1_{T-1}, w^2_0, \dots, w^2_{T-1} \}$ are independent and 
	have finite second moment. For convenience we have presented the result under the 
	additional assumption that the noise is zero-mean but that assumption can be relaxed 
	using a simple change of variables. 
\end{remark}

\section{Proof of the Main Result}
\subsection{Roadmap of the proof}
Our proof is based on the following fact which is typically referred to as
\emph{completion of squares} in the literature.

\begin{fact}
Given a linear system
$x_{t+1} = A x_t + B u_t + w_t$,
the quadratic cost 
\[
	\sum_{t=0}^{T-1} [x_t^\TRANS Q_t x_t 
	+ u_t^\TRANS R_t u_t] + x_T^\TRANS Q_T
	x_T 
\]
may be rewritten as
\[
	x_0^\TRANS P_0 x_0 + \sum_{t=0}^{T-1} (u_t + L_t x_t)^\TRANS \Delta_t (u_t + L_t x_t) + 
	\sum_{t=0}^{T-1} w_t^\TRANS P_{t+1} w_t,
\]
where $P_T = Q_T$ and for $t \in \{T-1, \dots, 0\}$, 
$P_t = \Riccati(P_{t+1}, A, B, Q_t, \mathbf{0}, R_t)$,
$L_t = \Gain(P_{t+1}, A, B,\mathbf{0}, R_t)$,  and
$\Delta_t = (R_t + B^\TRANS P_{t+1} B)$.
\end{fact}

Using this fact, one can prove the structure of optimal strategy for the
centralized control of stochastic linear systems for both complete and partial
state observation. See, for example,~\cite[Chapter 8]{Astrom1970}. However,
the completion of squares argument does not work directly for decentralized
control systems. 

In our proof, we exploit a fundamental 
property of the model, which was established in~\cite[Claim 2]{Asghari2018}
and is formally stated as Lemma~\ref{lem:indep}
below:
\(
  x^1_t \independent x^2_t | H^\com_t.
\)
As a consequence of this conditional
independence, the past realizations 
$(x^i_{0{:}t-1}, u^i_{0{:}t-1})$ are 
irrelevant at controller $C^i$ and may be shed without loss of 
optimality. This follows from Blackwell's 
principle of irrelevant information~\cite{Blackwell1964} as 
generalized to decentralized control systems 
in~\cite{Mahajan2015b}. The simplified structure of the optimal controller was established 
in~\cite[Lemma 1]{Asghari2018} and is formally
stated as Lemma~\ref{lem:structure} below. 

Using these two results and basic properties of conditional expectations, we
prove the structure of the dynamics of the state estimates and the estimation
error (Lemma~\ref{lem:estimates}). This structure was also established in~\cite[Theorem
3]{Asghari2018} as part of the result that establishes the structure of the
optimal controller. However, as we show below, one
only needs the conditional independence property of Lemma~\ref{lem:indep} and
its consequences to establish Lemma~\ref{lem:estimates}. 

As a next step, we use orthogonal projection and the specific form of the
information structure to simplify the per-step cost (Lemma~\ref{lem:cost}). We
combine this simplified form of the cost with the dynamics of the state
estimates and estimation error (established in Lemma~\ref{lem:estimates}) to
prove completion of squares result for the cost (Theorem~\ref{thm:cost-to-go})
tailored to the specific model of the system. 

Subsequently, we follow the standard steps of the ``completion of squares''
argument to establish the structure of the optimal strategy.

\subsection{Conditional independence of local states and its implications}

A key property of the model established in~\cite[Claim 2]{Asghari2018} is the
following.
\begin{lemma} \label{lem:indep}
  For any control strategy profile~$\mathbf g$ of the
  form~\eqref{eq:control},
  \[
    x^1_t \independent x^2_t \mid H^\com_t.
  \]
\end{lemma}

Furthermore, it is shown in~\cite[Lemma 1]{Asghari2018} that the above
conditional independence implies the following.
\begin{lemma} \label{lem:structure}
  In Problem~\ref{prob:main}, there is no loss of optimality to restrict
  attention to local controllers of the form
  \begin{equation} \label{eq:structure}
    u^i_t = g^i_t(x^i_t, H^\com_t), \quad i \in \{1, 2\}.
  \end{equation}
\end{lemma}

An immediate consequence of the above lemma is the following.
\begin{corollary} \label{cor:indep}
  For any control strategy profile~$\mathbf g$ of the
  form~\eqref{eq:structure}, we have the following:
  \begin{enumerate}
    \item $(x^1_t, u^1_t) \independent (x^2_t, u^2_t) \mid H^\com_t$.
    \item $(\tilde x^1_t, \tilde u^1_t) \independent (\tilde x^2_t,
      \tilde u^2_t) \mid H^\com_t$.
  \end{enumerate}
\end{corollary}
\begin{proof}
  Property 1 follows from the Lemma~\ref{lem:indep} and the structure of the
  control strategy. Property 2 follows from Property 1, Eqs.~\eqref{eq:x} and
  \eqref{eq:u} and the fact that $\hat x^i_t$ and $\hat u^i_t$ are functions
  of $H^\com_t$.
\end{proof}

\subsection{Some preliminary properties}
\begin{lemma}
  For any control strategy profile~$\mathbf g$ of the
  form~\eqref{eq:control}, we have the following:
  \begin{description}
    \item[\textbf{(H1)}] $\hat u^0_t = u^0_t$ and $\tilde u^0_t = \mathbf{0}$. Thus,
      $\hat u_t = \VVEC(u^0_t, \hat u^1_t, \hat u^2_t)$ and $\tilde u_t =
      \VVEC(\mathbf 0, \tilde u^1_t, \tilde u^2_t)$.
    \item[\textbf{(H2)}] $\EXP[ \tilde x_t \mid H^\com_t] = 0$ and
      $\EXP[\tilde u_t \mid H^\com_t] = 0$.
    \item[\textbf{(H3)}] For any matrix $W$ of appropriate dimensions,
      $\EXP[ \hat s_t^\TRANS W \tilde s_t ] = 0$, where
      $\hat s_t = \VVEC(\hat x_t, \hat u_t)$ and
      $\tilde s_t = \VVEC(\tilde x_t, \tilde u_t)$.
  \end{description}
  Furthermore, if the strategy profile is of the
  form~\eqref{eq:structure}, we have:
  \begin{description}
    \item[\textbf{(H4)}] For any matrix $W$ of appropriate dimensions,
      $\EXP[ (\tilde s^1_t)^\TRANS W \tilde s^2_t]=0$, where $\tilde s^i_t =
        \VVEC(\tilde x^i_t, \tilde u^i_t)$.
  \end{description}
\end{lemma}
\begin{proof}
  Property (H1) follows from the fact that $u^0_t$ is a measurable function of
  $H^0_t$ (which is the same as $H^\com_t$).

  Property (H2) is a standard
  property of error estimates and can be shown as follows:
  \[
    \EXP[ \tilde x_t \mid H^\com_t] =
    \EXP[ x_t - \EXP[ x_t \mid H^\com_t] \mid H^\com_t] = 0.
  \]

  Property (H3) follows from the generalized orthogonal principle and can be
  shown as follows:
  \[
    \EXP[ \hat s_t^\TRANS W \tilde s_t ] =
    \EXP\bigl[ \EXP[ \hat s_t^\TRANS W \tilde s_t \mid H^\com_t ] \bigr]
    =
    \EXP\bigl[ \hat s_t^\TRANS W
    \underbrace{\EXP[ \tilde s_t \mid H^\com_t ]}_
    {= 0 \text{ (by (H2))}} \bigr].
  \]

  Property (H4) follows from Corollary~\ref{cor:indep} and (H2).
\end{proof}

\subsection{Proof of Lemma~\ref{lem:estimates}}
\label{sec:estimates}

By definition, $H^\com_{t+1} = H^\com_{t} \cup \{ z^1_{t+1}, z^2_{t+1},
\Gamma^1_{t+1}, \Gamma^2_{t+1}, u^0_t \}$. Thus,
\begin{align}
  \hat x^i_{t+1} &= \EXP[ x^i_{t+1} \mid H^\com_{t+1}] \notag \\
  &= \EXP[ x^i_{t+1} \mid H^\com_t, z^1_{t+1}, z^2_{t+1},
  \Gamma^1_{t+1},\Gamma^2_{t+1}, u^0_t ] \notag \\
  &= \EXP[ x^i_{t+1} \mid H^\com_t, z^i_{t+1},
  \Gamma^i_{t+1}]  \label{eq:x-hat-equiv}
\end{align}
where we can remove $z^{-i}_{t+1}$ (where $-i$ means the controller other than~$i$),
$\Gamma^{-i}_{t+1}$, and $u^0_{t}$ due to the following reasons:
\begin{itemize}
  \item By~\eqref{eq:channel}, $z^{-i}_{t+1} = f(x^{-i}_{t+1},
    \Gamma^{-i}_{t+1})$ and hence conditionally independent of $x^i_{t+1}$
		given $H^\com_t$ due to (A1), Lemma~\ref{lem:indep} 
		and Lemma~\ref{lem:structure}.
    \item $\Gamma^{-i}_{t+1}$ is conditionally independent of $x^i_{t+1}$
      given $\{ H^\com_t, u^0_{t} \}$ due to~\eqref{eq:dynamics} and (A1).
    \item $u^0_t = g^0_t(H^0_t)$ and hence may be removed from the
      conditioning (since $H^\com_t = H^0_t$).
\end{itemize}

Now, we consider the two cases $\Gamma^i_{t+1} = 0$ and $\Gamma^i_{t+1} = 1$
separately. When $\Gamma^i_{t+1} = 0$, $z^i_{t+1} = \BLANK$ and
from~\eqref{eq:x-hat-equiv} we have
\begin{align*}
  \hat x^i_{t+1} &= \EXP[ x^i_{t+1} \mid H^\com_t, z^i_{t+1} = \BLANK,
  \Gamma^i_{t+1} = 0 ] \\
  &\stackrel{(a)}= \EXP[ x^i_{t+1} \mid H^\com_t ] \\
  &\stackrel{(b)}= A^{ii} \hat x^i_{t} + B^{i0} u^0_t + B^{ii} \hat u^i_t,
\end{align*}
where $(a)$ follows from (A1) and $(b)$ follows from~\eqref{eq:dynamics},
\eqref{eq:x}, \eqref{eq:u}, (H1), (A1), and (A2). Consequently,
\[
  \tilde x^i_{t+1} = x^i_{t+1} - \hat x^i_{t+1}
  = A^{ii} \tilde x^i_t + B^{ii} \tilde u^i_t + w^i_t.
\]

Now consider the case when $\Gamma^i_{t+1} = 1$, i.e., $z^i_{t+1} =
x^i_{t+1}$. Therefore,
\[
  \hat x^i_{t+1} = \EXP[ x^i_{t+1} \mid H^\com_{t}, z^i_{t+1} = x^i_{t+1},
  \Gamma^i_{t+1} = 1 ] = x^i_{t+1}.
\]
Consequently, $\tilde x^i_{t+1} = x^i_{t+1} - \hat x^i_{t+1} = 0$.

\subsection{Orthogonal projection for per-step cost}
\begin{lemma} \label{lem:cost}
  For any strategy profile of the form~\eqref{eq:structure}, we have
  \begin{align}
    \EXP[ x_t^\TRANS Q_t x_t ] &= \EXP\Bigl[ \hat x_t^\TRANS Q_t \hat x_t
      + \sum_{i \in \{1, 2\}} (\tilde x^{i}_t)^\TRANS Q_t^{ii} \tilde x^i_t
    \Bigr], \label{eq:x-cost}
    \\
    \EXP[ u_t^\TRANS R_t u_t ] &= \EXP\Bigl[ \hat u_t^\TRANS R_t \hat u_t
      + \sum_{i \in \{1, 2\}} (\tilde u^{i}_t)^\TRANS R_t^{ii} \tilde u^i_t
    \Bigr], \label{eq:u-cost}
    \\
    \EXP[ x_t^\TRANS M_t u_t ] &= \EXP\Bigl[ \hat x_t^\TRANS M_t \hat u_t
      + \sum_{i \in \{1, 2\}} (\tilde x^{i}_t)^\TRANS M_t^{ii} \tilde u^i_t
    \Bigr]. \label{eq:xu-cost}
  \end{align}
  Thus, we have that
  \begin{multline*}
    \EXP\Biggl[
      \MATRIX{ x_t \\ u_t }^\TRANS
      \MATRIX{ Q_t & M_t \\ M_t^\TRANS & R_t}
      \MATRIX{ x_t \\ u_t } \Biggr] =
    \EXP\Biggl[
      \MATRIX{ \hat x_t \\ \hat u_t }^\TRANS
      \MATRIX{ Q_t & M_t \\ M_t^\TRANS & R_t}
      \MATRIX{ \hat x_t \\ \hat u_t } \Biggr] \\
    + \sum_{i \in \{1, 2\}}
    \EXP\Biggl[
      \MATRIX{ \tilde x^i_t \\ \tilde u^i_t }^\TRANS
      \MATRIX{ Q^{ii}_t & M^{ii}_t \\ (M^{ii}_t)^\TRANS & R^{ii}_t}
      \MATRIX{ \tilde x^i_t \\ \tilde u^i_t } \Biggr]
  \end{multline*}
\end{lemma}
\begin{proof}
  To show~\eqref{eq:x-cost}, we recall that $x_t = \hat x_t + \tilde x_t$.
  Thus,
  \begin{equation} \label{eq:x-split}
    \EXP[ x_t^\TRANS Q_t x_t ] = \EXP[ \hat x_t^\TRANS Q_t \hat x_t
    + \tilde x_t^\TRANS Q_t \tilde x_t + 2 \hat x_t^\TRANS Q_t \tilde x_t].
  \end{equation}

  Consider the second term of~\eqref{eq:x-split}
  \begin{equation} \label{eq:x-tilde}
    \EXP[ \tilde x_t^\TRANS Q_t \tilde x_t ] =
    \sum_{i \in \{1, 2\}} \EXP[ (\tilde x^i_t)^\TRANS Q_t^{ii} \tilde x^i_t ]
    +
    2
    \underbrace{\EXP[ (\tilde x^1_t)^\TRANS Q_t^{12} \tilde x^2_t ]}
    _{ = 0 \text{ (by (H4))}}.
  \end{equation}
  Substituting~\eqref{eq:x-tilde} in~\eqref{eq:x-split} and observing that the
  third term of~\eqref{eq:x-split} is $0$ due to (H3), we
  get~\eqref{eq:x-cost}.

  Eqs.~\eqref{eq:u-cost} and~\eqref{eq:xu-cost} can be proved in a similar
  manner.
\end{proof}

\subsection{A change of variables}
For ease of notation, we define
\begin{align}
  \hat x^{i,\off}_{t+1} &= A^{ii} \hat x^i_t + B^{i0} u^0_t + B^{ii} \hat
  u^i_t,
  \\
  \tilde x^{i,\off}_{t+1} &= A^{ii} \tilde x^i_t +  B^{ii} \tilde
  u^i_t + w^i_t. \label{eq:x-tilde-off}
\end{align}
Thus, we can write
\[
    \hat x^i_{t+1} =
    \begin{cases}
      \hat x^{i,\off}_{t+1}, & \text{if 
      $\Gamma^i_{t+1} = 0$} \\
      x^i_{t+1}, & \text{if 
      $\Gamma^i_{t+1} = 1,$}
    \end{cases}
  \]
and 
\[
    \tilde x^i_{t+1} =
    \begin{cases}
      \tilde x^{i, \off}_{t+1}, & \text{if 
      $\Gamma^i_{t+1} = 0$} \\
      0, & \text{if 
      $\Gamma^i_{t+1} = 1.$}
    \end{cases}
  \]
Let $\hat x^\off_t = \VVEC(\hat x^{1,\off}_t, \hat x^{2,\off}_t)$ and $\tilde x^\off_t =
\VVEC(\tilde x^{1, \off}_t, \tilde x^{2, \off}_t)$.
It follows that $\hat x^{i,\off}_{t+1} + \tilde x^{i,\off}_{t+1} = x^i_{t+1}$
and
\begin{equation} \label{eq:x-hat-off}
  \hat x^\off_{t+1} = A \hat x_t + B \hat u_t.
\end{equation}

\begin{lemma}
  For any strategy profile of the form~\eqref{eq:control}, we have the
  following:
  \begin{description}
    \item[\textbf{(H5)}] For any matrix $W$ of appropriate dimensions,
      $\EXP[\hat x^\off_t W \tilde x^\off_t ] = 0$.
  \end{description}
  Furthermore, if the strategy profile is of the
  form~\eqref{eq:structure}, we have:
  \begin{description}
    \item[\textbf{(H6)}] For any matrix $W$ of appropriate dimensions,
      $\EXP[(\tilde x^{1,\off}_t)^\TRANS W \tilde x^{2,\off}_t ] = 0$.
  \end{description}
\end{lemma}
\begin{proof}
  Property (H5) follows immediately from (H3). Property (H6) follows
  immediately from (H4) and (A1).
\end{proof}

\begin{lemma}  \label{lem:rewrite}
  For any strategy profile of the form~\eqref{eq:structure}, we have the following:
  \begin{align}
    \hskip 1em & \hskip -1em
    \EXP\Bigl[ \hat x_{t+1}^\TRANS P_{t+1} \hat x_{t+1} +
      \smashoperator{\sum_{i \in \{1, 2\}}}
      (\tilde x^i_{t+1})^\TRANS \tilde P^i_{t+1}
      \tilde x^i_{t+1}
    \Bigr] \notag \\
    &= \EXP\Bigl[ (\hat x^\off_{t+1})^\TRANS P_{t+1} \hat x^\off_{t+1} +
          \smashoperator{\sum_{i \in \{1, 2\}}}
          (\tilde x^{i,\off}_{t+1})^\TRANS
          \Pi^i_{t+1} \tilde x^{i,\off}_{t+1}
        \Bigr] . \label{eq:update}
  \end{align}
\end{lemma}
\begin{proof}
  We compute the conditional value of the left hand side given the realization
  of $\Gamma_{t+1} = (\Gamma^1_{t+1}, \Gamma^2_{t+1})$ and using
  Lemma~\ref{lem:estimates}. We have four cases
  \begin{enumerate}
    \item $\Gamma_{t+1} = (0, 0)$: In this case $\hat x_{t+1} = \hat x^\off_{t+1}$
      and $\tilde x_{t+1} = \tilde x^\off_{t+1}$. Thus,
      \begin{align*}
        &\EXP\Bigl[ \hat x_{t+1}^\TRANS P_{t+1} \hat x_{t+1} +
          \smashoperator{\sum_{i \in \{1, 2\}}}
          (\tilde x^i_{t+1})^\TRANS \tilde P^i_{t+1}
          \tilde x^i_{t+1}
          \Big| \Gamma_{t+1} = (0,0)
        \Bigr] \\
        &\quad=
        \EXP\Bigl[ (\hat x^\off_{t+1})^\TRANS P_{t+1} \hat x^\off_{t+1} +
          \smashoperator{\sum_{i \in \{1, 2\}}}
          (\tilde x^{i,\off}_{t+1})^\TRANS \tilde
          P^i_{t+1} \tilde x^{i,\off}_{t+1}
        \Bigr].
      \end{align*}
    \item $\Gamma_{t+1} = (1, 0)$: In this case $\hat x_{t+1} = \VVEC(x^1_{t+1},
      \hat x^{2,\off}_{t+1}) = \hat x^\off_{t+1} + \VVEC(\tilde x^{1, \off}_{t+1},
      \mathbf 0)$ and $\tilde x_{t+1} = \VVEC(\mathbf 0, \tilde
      x^{2,\off}_{t+1})$. Thus,
      \begin{align*}
        &\EXP\Bigl[ \hat x_{t+1}^\TRANS P_{t+1} \hat x_{t+1} +
          \smashoperator{\sum_{i \in \{1, 2\}}}
          (\tilde x^i_{t+1})^\TRANS \tilde P^i_{t+1}
          \tilde x^i_{t+1}
          \Big| \Gamma_{t+1} = (1,0)
        \Bigr] \\
        &\quad =
        \EXP\Bigl[ (\hat x^\off_{t+1})^\TRANS P_{t+1} \hat x^\off_{t+1} +
          (\tilde x^{1,\off}_{t+1})^\TRANS  P^{11}_{t+1}
          \tilde x^{1,\off}_{t+1} \\
        &\hskip 10.7em +
          (\tilde x^{2,\off}_{t+1})^\TRANS  \tilde P^{2}_{t+1}
          \tilde x^{2,\off}_{t+1}
        \Bigr].
      \end{align*}
    \item $\Gamma_{t+1} = (0,1)$: Similar to case 2), we can show that
      \begin{align*}
        &\EXP\Bigl[ \hat x_{t+1}^\TRANS P_{t+1} \hat x_{t+1} +
          \smashoperator{\sum_{i \in \{1, 2\}}}
          (\tilde x^i_{t+1})^\TRANS \tilde P^i_{t+1}
          \tilde x^i_{t+1}
          \Big| \Gamma_{t+1} = (0,1)
        \Bigr] \\
        &\quad =
        \EXP\Bigl[ (\hat x^\off_{t+1})^\TRANS P_{t+1} \hat x^\off_{t+1} +
          (\tilde x^{1,\off}_{t+1})^\TRANS  \tilde P^{1}_{t+1}
          \tilde x^{1,\off}_{t+1} \\
        &\hskip10.7em +
          (\tilde x^{2,\off}_{t+1})^\TRANS  P^{22}_{t+1}
          \tilde x^{2,\off}_{t+1}
        \Bigr].
      \end{align*}
    \item $\Gamma_{t+1} = (1,1)$: In this case, $\hat x_{t+1} = x_{t+1} = \hat
      x^{\off}_{t+1} + \tilde x^\off_{t+1}$ and $\tilde x_{t+1} = \mathbf 0$.
      Thus,
      \begin{align*}
        &\EXP\Bigl[ \hat x_{t+1}^\TRANS P_{t+1} \hat x_{t+1} +
          \smashoperator{\sum_{i \in \{1, 2\}}}
          (\tilde x^i_{t+1})^\TRANS \tilde P^i_{t+1}
          \tilde x^i_{t+1}
          \Big| \Gamma_{t+1} = (1,1)
        \Bigr] \\
        &= \EXP\Bigl[ (\hat x^\off_{t+1})^\TRANS P_{t+1} \hat x^\off_{t+1}
          + (\tilde x^\off_{t+1})^\TRANS P_{t+1} \tilde x^\off_{t+1} \notag\\
        & \hskip 9.7em
          + 2 (\hat x^\off_{t+1})^\TRANS P_{t+1} \tilde x^\off_{t+1}
        \Bigr] \\
        &=
        \EXP\Bigl[ (\hat x^\off_{t+1})^\TRANS P_{t+1} \hat x^\off_{t+1} +
          \smashoperator{\sum_{i \in \{1, 2\}}}
          (\tilde x^{i,\off}_{t+1})^\TRANS
          P^{ii}_{t+1} \tilde x^{i,\off}_{t+1}
        \Bigr],
      \end{align*}
      where the last equality follows from (H5) and (H6).
  \end{enumerate}

  Combining these four cases and using the law of total probability, we
  get~\eqref{eq:update}.
\end{proof}

\subsection{Completion of squares}
\begin{lemma} \label{lem:square}
  Let $x \in \X{}$, $u \in \U{}$, and $w \in \X{}$ be random variables defined
  on a common probability space. Suppose $w$ is zero mean with finite covariance 
  and independent of
  $(x,u)$. Let $x_{+} = Ax + Bu + w$, where $A$ and $B$ are matrices of
  appropriate dimensions. Then given matrices $P$, $Q$, $M$, and $R$ of
  appropriate dimensions,
  \begin{multline*}
    \EXP\Biggl[ \MATRIX{x \\ u}^\TRANS \MATRIX{ Q & M \\ M^\TRANS & R}
    \MATRIX{ x \\ u } + x_+^\TRANS P x_+ \Biggr]
    \\
    = \EXP\bigl[x^\TRANS P_{+} x + (u + Kx)^\TRANS \Delta (u + Kx)
      + w^\TRANS P w \bigr],
  \end{multline*}
  where
  \begin{align*}
    \Delta &= R + B^\TRANS P B , \\
    K &= \Delta^{-1} [ M^\TRANS + B^\TRANS P A ], \\
    P_{+} &= Q + A^\TRANS P A - K^\TRANS \Delta K.
  \end{align*}
\end{lemma}
\begin{proof}
  Since $w$ is zero mean and independent of $(x,u)$, we have
  \[
    \EXP[ x_+^\TRANS P x_+] =
    \EXP\bigl[ (Ax + Bu)^\TRANS P (Ax+Bu) + w^\TRANS P w \bigr]
  \]
  The result follows by expanding both sides and comparing coefficients.
\end{proof}

By combining Lemmas~\ref{lem:cost}, \ref{lem:rewrite} and~\ref{lem:square},
we get the following.
\begin{lemma} \label{lem:square-step}
  For any strategy profile of the form~\eqref{eq:structure},
  \begin{align*}
    \EXP&\Bigl[  c_t(x_t, u_t) + \hat x_{t+1}^\TRANS P_{t+1} \hat x_{t+1}
      + \sum_{i \in \{1, 2\}} (\tilde x^i_{t+1})^\TRANS \tilde P^i_{t+1}
    \tilde x^i_{t+1} \Bigr] \notag \\
    &= \EXP\Bigl[ \hat x_t^\TRANS P_t \hat x_t +
        (\hat u_t + K_t \hat x_t)^\TRANS \Delta_t (\hat u_t + K_t \hat x_t)
    \notag \\
    &\qquad + \smashoperator[l]{\sum_{i \in \{1,2\}}}
    \Bigl[ (\tilde x^i_t)^\TRANS \tilde P^i_t \tilde x^i_t +
        (\tilde u_t + \tilde K_t \tilde x_t)^\TRANS \tilde \Delta_t (\tilde u_t + \tilde K_t \tilde
        x_t) \notag \\
    &\hskip 5em + (w^i_t)^\TRANS \tilde \Pi^i_{t+1} w^i_t \Bigr]\Bigr].
  \end{align*}
\end{lemma}
\begin{theorem} \label{thm:cost-to-go}
  For any strategy profile $\mathbf g$ of the form~\eqref{eq:structure},
  \begin{align}
  J(\mathbf g) = \EXP^{\mathbf g}&\Bigl[ \hat x_0^\TRANS P_{t} \hat x_0 +
    \sum_{i \in \{1, 2\}} (\tilde x^i_0)^\TRANS \tilde P^i_0 \tilde x^i_0
    \notag \\
    & + \sum_{s=0}^{T-1}
      (\hat u_s + K_s \hat x_s)^\TRANS \Delta_s (\hat u_s + K_s \hat x_s)
    \notag\\
    &  + \sum_{s=0}^{T-1} \sum_{i \in \{1, 2\}}
    (\tilde u^i_s + \tilde K^i_s \tilde x^i_s)^\TRANS
    \tilde \Delta^i_s (\tilde u^i_s + \tilde K^i_s \tilde x^i_s)
    \notag\\
    & + \sum_{s=0}^{T-1} \sum_{i \in \{1, 2\}} (w^i_s)^\TRANS \Pi^i_{t+1}
    w^i_s
  \Bigr], \label{eq:cost-to-go}
  \end{align}
  where
  \(
    \Delta_s = R_s + B^\TRANS P_{s+1} B
  \)
  and
  \(
    \tilde \Delta^i_s = R^{ii}_s + (B^{ii})^\TRANS \Pi^{i}_{s+1} B^{ii}
  \),
  \(
     i \in \{1,2\}.
  \)
\end{theorem}

\begin{proof}
For any strategy profile~$\mathbf g$, define the expected cost to go from time~$t$ onwards as
\begin{equation}
  V_t(\mathbf g) =
  \EXP^{\mathbf g} \Bigl[ \sum_{s = t}^{T-1} c_s(x_s, u_s) + c_T(x_T) \Bigr].
\end{equation}
We claim that 
  \begin{align}
    V_t(\mathbf g) = \EXP^{\mathbf g}&\Bigl[ \hat x_t^\TRANS P_{t} \hat x_t +
      \sum_{i \in \{1, 2\}} (\tilde x^i_t)^\TRANS \tilde P^i_t \tilde x^i_t
      \notag \\
      & + \sum_{s=t}^{T-1}
        (\hat u_s + K_s \hat x_s)^\TRANS \Delta_s (\hat u_s + K_s \hat x_s)
      \notag\\
      &  + \sum_{s=t}^{T-1} \sum_{i \in \{1, 2\}}
      (\tilde u^i_s + \tilde K^i_s \tilde x^i_s)^\TRANS
      \tilde \Delta^i_s (\tilde u^i_s + \tilde K^i_s \tilde x^i_s)
      \notag\\
      & + \sum_{s=t}^{T-1} \sum_{i \in \{1, 2\}} (w^i_s)^\TRANS \Pi^i_{t+1}
      w^i_s
    \Bigr]. \label{eq:claim}
  \end{align}

  We prove the claim by backward induction. For $t = T$, Lemma~\ref{lem:cost}
  implies that
  \[
    V_T(\mathbf g) = \EXP\Bigl[ \hat x_T^\TRANS Q_T \hat x_T +
    \sum_{i \in \{1, 2\}} (\tilde x^i_T)^\TRANS Q^{ii}_T \tilde x^i_T \Bigr].
  \]
  Eq.~\eqref{eq:claim} follows from the definition of $P_T$ and $\tilde P^i_T$.
  This forms the basis of induction. Now assume that~\eqref{eq:claim} is true for $t+1$
  and consider $V_t$. By definition, we have  
  \[
    V_t(\mathbf g) = \EXP^{\mathbf g}[ c_t(x_t, u_t) ] + V_{t+1}(\mathbf g)
  \]
  Using the expression for $V_{t+1}$ and  Lemma~\ref{lem:square-step}, we get the expression
  for $V_t$. This completes the induction step and proves the claim~\eqref{eq:claim}. 
  
  The result of the Theorem then follows from observing that $J(\mathbf g) = V_0(\mathbf g)$.
\end{proof}

\subsection{Proof of Theorem~\ref{thm:optimal}}
By Lemma~\ref{lem:structure}, there is no loss of optimality in restricting
attention to control strategy profile of the form~\eqref{eq:structure}. By Theorem~\ref{thm:cost-to-go},
the performance of a strategy of the form~\eqref{eq:structure} is given by~\eqref{eq:cost-to-go}.
Note that the first two and the last terms of~\eqref{eq:cost-to-go} are control free (i.e., they depend on only primitive random
variables). Thus, minimizing $J(\mathbf g)$ is equivalent to minimizing
\begin{align*}
  \tilde J(\mathbf g) = \EXP^{\mathbf g}\Bigl[
     \sum_{s=0}^{T-1} \Bigl[&
      (\hat u_s + K_s \hat x_s)^\TRANS \Delta_s (\hat u_s + K_s \hat x_s)
    \notag\\
    &  + \sum_{i \in \{1, 2\}}
    (\tilde u^i_s + \tilde K^i_s \tilde x^i_s)^\TRANS
    \tilde \Delta^i_s (\tilde u^i_s + \tilde K^i_s \tilde x^i_s)
\Bigr] \Bigr].
\end{align*}

By (A3), $R_t$ is symmetric and positive definite and therefore so is
$R^{ii}_t$. It can be shown recursively that $P_t$ and $\tilde P_t$ are
symmetric and positive semi-definite. Hence both $\Delta_t$ and $\tilde
\Delta^{i}_t$ are symmetric and positive definite. Therefore, $\tilde
J(\mathbf g) \ge 0$ with equality if and only if the strategy profile $\mathbf
g$ is given by Theorem~\ref{thm:optimal}.

\section{Discussion}

The model in~\cite{Asghari2018} consisted of $N$ local controllers and one
remote controller. We restricted our discussion to $N=2$. All steps of our
proof apart from Lemma~\ref{lem:rewrite} extend trivially to the case of
general~$N$. To extend Lemma~\ref{lem:rewrite} to the case of general $N$,
one can establish the following result.
\begin{lemma} \label{lem:generalize}
  For the system with general $N$, for any $\gamma = (\gamma^1, \dots,
  \gamma^n)$, $\gamma^i \in \{ 0,1 \}$, we have
  \begin{multline*}
    \EXP\Bigl[\hat x_{t+1}^\TRANS P_{t+1} \hat x_{t+1} + \sum_{i =1}^{N}
    (\tilde x^i_{t+1})^\TRANS \tilde P^i_{t+1} \tilde x^i_{t+1} \Bigm| \Gamma_{t+1} = \gamma\Bigr] \\
    = \EXP\Bigl[(\hat x^{\off}_{t+1})^\TRANS P_{t+1} \hat x^{\off}_{t+1}
      \hskip 8em \\
      + \sum_{i =1}^{N}
    (\tilde x^{i,\off}_{t+1})^\TRANS \Lambda^i_{t+1}(\gamma^i) \tilde
  x^{i,\off}_{t+1} \Bigm| \Gamma_{t+1} = \gamma\Bigr],
  \end{multline*}
where $\Lambda^i_{t+1}(\gamma^i) = \begin{cases}
\tilde P^i_{t+1}, & \text{if } \gamma^i = 0\\
P^{ii}_{t+1}  & \text{if } \gamma^i = 1.
\end{cases}$
\end{lemma}

Lemma~\ref{lem:rewrite} then follows from observing that
\begin{align*}
  &\sum_{(\gamma^1, \dots, \gamma^N) \in \{ 0,1\}^N } \PR(\Gamma^1_t =
  \gamma^1) \cdots \PR(\Gamma^N_t = \gamma^N) \Lambda^i_{t+1}(\gamma^i)\\
  &\qquad = \sum_{\gamma^i \in \{ 0,1\} } \PR(\Gamma^i_t = \gamma^i)  \Lambda^i_{t+1}(\gamma^i)\\
  &\qquad = p^i \tilde P^i_{t+1} + (1-p^i) P^{ii}_{t+1} \\
  &\qquad = \Pi^i_{t+1}.
\end{align*}

The proof of Theorem~\ref{thm:cost-to-go} is similar in spirit to the proof of
centralized linear quadratic control presented in~\cite{Astrom1970}. However,
due to decentralized information and the presence of unreliable communication
channels, the specific details are different. As far as we are aware, this is
the first paper which presents a methodology to synthesize optimal controllers
for dynamic \emph{decentralized} control systems without using a dynamic
programming or a spectral decomposition argument. In contrast to dynamic
programming based approaches, we sidestep the subtle measurability issues that
arise in common information based dynamic program for continuous state and
action spaces. In contrast to spectral decomposition based arguments, we do
not apriori restrict attention to linear strategies. We believe that the
solution approach presented in this paper is interesting in its own right and
may be applicable to other decentralized control problems as well.

\section*{Acknowledgment}

The authors are grateful to Ashutosh Nayyar and Yi Ouyang for helpful
feedback.

\bibliography{IEEEabrv,../../../References/mybib}
\bibliographystyle{IEEEtran}

\end{document}